\documentclass[a4paper, 10 pt, conference]{ieeeconf} 
\IEEEoverridecommandlockouts
\usepackage[top=104pt, bottom=54pt,left=54pt,right=37pt]{geometry}

\usepackage{amsfonts,amsmath,amssymb} 

\newcommand{\dif}{\,\mathrm{d}}		  

\usepackage{psfrag,color}
\usepackage{enumerate,cite,latexsym,graphicx}
\newtheorem{theorem}{Theorem}
\newtheorem{lemma}{Lemma}

\newtheorem{definition}{Definition}

\newtheorem{remark}{Remark}

\title{\LARGE \bf Singular Perturbation Approximations for General Linear Quantum 
Systems
}

\author{Shanon L.~Vuglar, Ian R.~Petersen%
\thanks{Shanon L. Vuglar is with the School of Engineering and Information 
Technology, 
        University of New South Wales at the Australian Defence Force Academy, Canberra ACT 2600, Australia.
         {\tt\small shanonvuglar@vuglar.com} }%
\thanks{Ian R. Petersen is with the School of Information Technology and Electrical Engineering, 
        University of New South Wales at the Australian Defence Force Academy, Canberra ACT 2600, Australia.
         {\tt\small i.r.petersen@gmail.com} } }%

\begin{document}

\maketitle
\thispagestyle{empty}
\pagestyle{empty}

\begin{abstract}
This paper considers the use of singular perturbation approximations for 
general linear quantum systems where the system dynamics are described in 
terms of both annihilation and creation operators. Results that are related to 
the physical realizability property of the approximate system are presented. 
\end{abstract}

\section{Introduction} \label{sec:intro}
Quantum feedback control is an active research area addressing the need to 
take into account quantum effects in systems that are inherently quantum in 
nature or when levels of accuracy approach the quantum noise limit. Examples 
where these effects need to be considered occur in quantum optics, quantum 
communications, quantum computing and precision measurement, for example 
gravity wave detection \cite{JNP1,MaP3,MaP4,DP3,GJN10,GJ09,NJD09,NJP1,GGY08,
MAB08,YNJP1,YAM06,ShP5a,PET09A,PET10Ca,GNW12,GN12,GV07,BVS08}.

Recent papers (for example \cite{JNP1,MaP3,MaP4,ShP5a}) use a consistent 
approach to describe quantum systems, 
(both plant and controller) as a combination of quantum harmonic oscillators 
coupled to quantum fields and described by quantum stochastic differential 
equations. The concept of physical realizability has been well defined 
(see \cite{JNP1}) and 
relates to whether a given synthesized system could be implemented as such a 
combination of quantum harmonic oscillators. Equivalent necessary and 
sufficient conditions for physical realizability are given in \cite{JNP1} 
and \cite{ShP5a}. 

Physical realizability is particularly relevant in both modeling and control 
applications. In the case of modeling a plant (for example as part of a 
controller design process), it is often necessary to have a plant model that 
is physically realizable as this can have implications for controller design. 
Similarly, in the case of coherent quantum feedback control where the 
controller is implemented as a quantum system, it is necessary to establish 
whether it is in fact possible to implement a synthesized controller as a 
quantum system.

Examples of components occurring in physical quantum systems that are 
adequately described by this framework include beam splitters, phase-shift 
modulators and optical cavities. In \cite{PET09A}, the physical realizability 
of singular perturbation approximations for a class of linear quantum systems 
is considered. Singular perturbation approximations are closely related to 
adiabatic elimination, a commonly used technique used in modeling quantum 
systems within the physics literature. 
Related results in a more general setting can be found in 
\cite{GV07,BVS08,GNW12,GN12}.
In \cite{PET09A}, the class of quantum linear 
systems that can be describe solely in terms of annihilation operators is 
considered. This class corresponds to passive systems, for example optical 
systems containing only passive components such as optical cavities, 
beam-splitters and phase shifters.

In light of the results obtained in \cite{PET09A}, this paper considers the 
more general class of quantum linear systems described by both annihilation and 
creation operators. Two results are obtained; one for a general singular 
perturbation, and one for a special case in which the Hamiltonian and 
coupling operators are singularly perturbed. In the general case a result 
(relevant to physical realizability) relating to the J-J unitary property of 
the transfer function of the approximate system is obtained. In the special 
case, it is shown that while in general, the system obtained from the 
singular perturbation approximation is not necessarily physically realizable, 
it is equivalent to a physically realizable system in series with a 
static Bogoliubov component (generalized static squeezer).

Components implementing static linear transformations, called 
Bogoliubov transformations (see \cite{GJN10}), such as a static squeezer, as 
encountered in quantum optics, do not belong to the class of physically 
realizable quantum systems, although they can be approximated by systems that 
are. In (\cite{GJN10}) a framework is developed to merge these to classes: 
(i) dynamical components, with linear evolution of physical variables, and 
(ii) static components characterized by Bogoliubov transformations. 
In particular, general methods for cascade, series, and feedback connections 
are provided, input-output maps and transfer functions for representing 
components are defined and the issue of convergence is addressed.

The remainder of the paper proceeds as follows. In Section~\ref{sec:model}, 
we describe the quantum system model used throughout this paper. 
In Section~\ref{sec:result}, we present our two main results related to the 
physical realizability property of approximate systems obtained through the 
use of singular perturbation approximations. The first result is for a 
system obtained via a general singular perturbation, whereas the second 
result is for a special case in which the Hamiltonian and coupling operators 
are singularly perturbed. 
An illustrative example follows in Section~\ref{sec:ex} and our conclusion is 
given in Section~\ref{sec:conc}.

\section{Quantum System Model} \label{sec:model}
To aid the reader, the nomenclature and variables used throughout this paper 
are consistent with their usage in \cite{PET09A}. $(.)^\#$ applied to an 
operator, is its operator adjoint, and applied to a matrix is its 
complex conjugate. $(.)^\dagger = (.)^{\#T}$.  

As in \cite{ShP5a},
we consider the class of linear quantum systems models representing 
$n$ quantum harmonic oscillators coupled to $m$ external independent quantum 
fields. This class of linear quantum systems is described by 
quantum stochastic differential equations (QSDEs) of the form
\begin{eqnarray}
\begin{bmatrix} \dif a(t) \\ \dif a(t)^\# \end{bmatrix}
&=& F \begin{bmatrix} a(t) \\ a(t)^\# \end{bmatrix} dt 
+ G \begin{bmatrix} \dif u(t) \\ \dif u(t)^\# 
	\end{bmatrix}; \nonumber \\
\begin{bmatrix}
	\dif y(t) \\
	\dif y(t)^\# 
\end{bmatrix}
&=& H \begin{bmatrix} a(t) \\ a(t)^\# \end{bmatrix} dt 
+ K \begin{bmatrix}
	\dif u(t) \\
	\dif u(t)^\# 
\end{bmatrix} \label{eqn:genModel}
\end{eqnarray} 
where $a(t) = \begin{bmatrix}a_1(t) & \cdots & a_n(t)\end{bmatrix}^T$
is a column vector of linear combinations of annihilation and creation 
operators corresponding to the harmonic oscillators. The vector $u(t)$ 
represents the input to the system. It is assumed to admit the decomposition 
$\dif u(t) = \beta_u(t) \dif t + \dif \tilde{u}(t)$ where $\tilde{u}(t)$
is the noise part of $u(t)$ (with Ito products
$\dif \tilde{u}(t) \dif \tilde{u}^T(t) = F_{\tilde{u}} \dif t$
where $F_{\tilde{u}}$ is non-negative Hermitian)
and $\beta_u(t)$ is the adapted, self adjoint part of $u(t)$. 
$F$, $G$, $H$, $K$ are of the form:
\begin{eqnarray*}
	F &=&  \begin{bmatrix} F_1 & F_2 \\ F_2^\# & F_1^\# \end{bmatrix}
		\in \mathbb{C}^{2n \times 2n}; \\
	G &=&  \begin{bmatrix} G_1 & G_2 \\ G_2^\# & G_1^\# \end{bmatrix}
		\in \mathbb{C}^{2n \times 2m}; \\
	H &=&  \begin{bmatrix} H_1 & H_2 \\ H_2^\# & H_1^\# \end{bmatrix}
		\in \mathbb{C}^{2m \times 2n}; \\
	K &=&  \begin{bmatrix} K_1 & K_2 \\ K_2^\# & K_1^\# \end{bmatrix}
		\in \mathbb{C}^{2m \times 2m }. \\
\end{eqnarray*}

\begin{definition} (see \cite{ShP5a}) A linear quantum system of the form 
(\ref{eqn:genModel}) is \emph{physically realizable} if there exists 
a commutation matrix $\Theta = \Theta^\dagger = TJT^\dagger \ge 0$, a 
Hamiltonian matrix $M$, coupling matrix $N$ and scattering matrix $S$, with 
\begin{eqnarray*}
	J &=& \begin{bmatrix} I & 0 \\ 0 & -I \end{bmatrix}; \\
	T &=& \begin{bmatrix} T_1 & T_2 \\
		T_2^\# & T_1^\# \end{bmatrix};
	\qquad \mbox{T non-singular}; \\
	M &=& \begin{bmatrix} M_1 & M_2 \\
		M_2^\# & M_1^\# \end{bmatrix}; 
	\qquad M^\dagger = M; \\
	N &=& \begin{bmatrix} N_1 & N_2 \\
		N_2^\# & N_1^\# \end{bmatrix}; 
	\qquad	S^{-1} = S^\dagger ,
\end{eqnarray*} 
such that,
\begin{eqnarray}
	F &=& -i\Theta M - \frac{1}{2} \Theta N^\dagger JN; \nonumber \\
	G &=& - \Theta N^\dagger JK; \nonumber \\
	H &=& N; \nonumber \\
	K &=& \begin{bmatrix} S & 0 \\ 0 & S^\#
	\end{bmatrix}. \label{eqn:realizable}
\end{eqnarray}
\end{definition}

\begin{theorem}
	(see \cite{ShP5a}) Suppose the linear quantum system 
	(\ref{eqn:genModel}) 
	is minimal, and that $\lambda_i(F) + \lambda_j(F) \ne 0$ for all 
	eigenvalues $\lambda_i(F)$, $\lambda_j(F)$ of $F$. Then this 
	linear quantum system is physically realizable if and only if 
	the following conditions hold:
	\begin{enumerate}
		\item The system transfer function matrix 
			$\Gamma(s)$ is (J,J)-unitary. That is, \newline
			$\Gamma^\sim(s)J\Gamma (s) = J$
			for all $s \in \mathbb{C}$; where 
			 $\Gamma^\sim(s) = \Gamma^\dagger(-s^*)$ and $s^*$ is 
			 the complex conjugate of $s$.
		\item $K$ is of the form
			$K = \begin{bmatrix} S & 0 \\ 0 & S^\#	\end{bmatrix}$ 
			where $S^\dagger S = S S^\dagger = I$.
	\end{enumerate}
\end{theorem}

\begin{definition}
	(see \cite{GJN10}) A \emph{static Bogoliubov component} is a 
	component that implements the \emph{Bogoliubov} transformation:
	$$\begin{bmatrix} \dif y(t) \\ \dif y^\#(t) \end{bmatrix} = B 
		\begin{bmatrix} \dif u(t) \\ \dif u^\#(t) \end{bmatrix};$$
	where
	$$B = \begin{bmatrix} B_1 & B_2 \\
	B_2^\# & B_1^\# \end{bmatrix};
	\qquad JB^\dagger J B = B J B^\dagger J = I.$$
\end{definition}

\begin{remark}A static Bogoliubov component is in general not physically 
	realizable in the above sense. However, it is a useful idealization 
	for certain devices used in quantum optics, (e.g. a static squeezer) 
	and is correctly interpreted as a limiting situation \cite{GJN10}.
\end{remark}

\section{Main Result}
\label{sec:result}
\subsection{General Singular Perturbations}
Consider the class of quantum systems of the form (\ref{eqn:genModel}) that 
are dependent on a parameter $\varepsilon \ge 0 $ as follows:
\begin{eqnarray}
\begin{bmatrix} \dif a_1(t) \\ \dif a_2(t) \\
	\dif a_1(t)^\# \\ \dif a_2(t)^\# \end{bmatrix}
&=& 
\underbrace{
\begin{bmatrix}
	F_{1a} & F_{1b} & F_{2a} & F_{2b} \\
	\frac{1}{\varepsilon}
	F_{1c} & 
	\frac{1}{\varepsilon}
	F_{1d} & 
	\frac{1}{\varepsilon}
	F_{2c} & 
	\frac{1}{\varepsilon}
	F_{2d} \\
	F_{1a}^\# & F_{1b}^\# & F_{2a}^\# & F_{2b}^\# \\
	\frac{1}{\varepsilon}
	F_{1c}^\# & 
	\frac{1}{\varepsilon}
	F_{1d}^\# & 
	\frac{1}{\varepsilon}
	F_{2c}^\# & 
	\frac{1}{\varepsilon}
	F_{2b}^\#
\end{bmatrix}
}_{F_\varepsilon}
\begin{bmatrix} a_1(t) \\ a_2(t) \\
	a_1(t)^\# \\ a_2(t)^\# \end{bmatrix}
dt \nonumber \\
&& \qquad +
\underbrace{
\begin{bmatrix}
	G_{1a} & G_{2a} \\
	\frac{1}{\varepsilon}
	G_{1b} & 
	\frac{1}{\varepsilon}
	G_{2b} \\
	G_{2a}^\# & G_{1a}^\# \\
	\frac{1}{\varepsilon}
	G_{2b}^\# & 
	\frac{1}{\varepsilon}
	G_{1b}^\# 
\end{bmatrix}
}_{G_\varepsilon}
\begin{bmatrix}
	\dif u(t) \\
	\dif u(t)^\# 
\end{bmatrix};
\nonumber \\
\begin{bmatrix}
	\dif y(t) \\
	\dif y(t)^\# 
\end{bmatrix}
&=& 
\underbrace{
\begin{bmatrix}
	H_{1a} & H_{1b} & H_{2a} & H_{2b} \\
	H_{1a}^\# & H_{1b}^\# & H_{2a}^\# & H_{2b}^\# \\
\end{bmatrix}}_H
\begin{bmatrix} a_1(t) \\ a_2(t) \\
	a_1(t)^\# \\ a_2(t)^\# \end{bmatrix}
dt \nonumber \\
&& \qquad +
\underbrace{
\begin{bmatrix}
	K_{1} & K_{2} \\
	K_{2}^\# & K_{1}^\# \\
\end{bmatrix}}_K
\begin{bmatrix}
	\dif y(t) \\
	\dif y(t)^\# 
\end{bmatrix}.
\label{eqn:sys}
\end{eqnarray} 

Equivalently, re-ordering partitions for convenience with new matrices as 
labeled,  and re-writing in the more 
standard singularly perturbed form we obtain:
\begin{eqnarray}
\begin{bmatrix} \dif a_1(t) \\
	\dif a_1(t)^\# 
\end{bmatrix}
&=& 
\underbrace{
\begin{bmatrix}
	F_{1a} & F_{2a} \\
	F_{1a}^\# & F_{2a}^\# \\
\end{bmatrix}
}_{F_{a}}
\begin{bmatrix} a_1(t) \\
	a_1(t)^\# 
\end{bmatrix}
dt \nonumber \\
&& \quad +
\underbrace{
\begin{bmatrix}
	F_{1b} & F_{2b} \\
	F_{1b}^\# & F_{2b}^\# \\
\end{bmatrix}
}_{F_{b}}
\begin{bmatrix} a_2(t) \\
	a_2(t)^\# 
\end{bmatrix}
dt \nonumber \\
&& \quad +
\underbrace{
\begin{bmatrix}
	G_{1a} & G_{2a} \\
	G_{2a}^\# & G_{1a}^\# \\
\end{bmatrix}
}_{G_{a}}
\begin{bmatrix}
	\dif u(t) \\
	\dif u(t)^\# 
\end{bmatrix};
\nonumber \\
\varepsilon 
\begin{bmatrix} \dif a_2(t) \\
	\dif a_2(t)^\# 
\end{bmatrix}
&=& 
\underbrace{
\begin{bmatrix}
	F_{1c} & F_{2c} \\
	F_{1c}^\# & F_{2c}^\# \\
\end{bmatrix}
}_{F_{c}}
\begin{bmatrix} a_1(t) \\
	a_1(t)^\# 
\end{bmatrix}
dt \nonumber \\
&& \quad +
\underbrace{
\begin{bmatrix}
	F_{1d} & F_{2d} \\
	F_{1d}^\# & F_{2d}^\# \\
\end{bmatrix}
}_{F_{d}}
\begin{bmatrix} a_2(t) \\
	a_2(t)^\# 
\end{bmatrix}
dt \nonumber \\
&& \quad +
\underbrace{
\begin{bmatrix}
	G_{1b} & G_{2b} \\
	G_{2b}^\# & G_{1b}^\# \\
\end{bmatrix}
}_{G_{b}}
\begin{bmatrix}
	\dif u(t) \\
	\dif u(t)^\# 
\end{bmatrix};
\nonumber \\
\begin{bmatrix}
	\dif y(t) \\
	\dif y(t)^\# 
\end{bmatrix}
&=& 
\underbrace{
\begin{bmatrix}
	H_{1a} & H_{2a} \\
	H_{1a}^\# & H_{2a}^\# \\
\end{bmatrix}
}_{H_{a}}
\begin{bmatrix} a_1(t) \\
	a_1(t)^\# 
\end{bmatrix}
dt \nonumber \\
&& \quad +
\underbrace{
\begin{bmatrix}
	H_{1b} & H_{2b} \\
	H_{1b}^\# & H_{2b}^\# \\
\end{bmatrix}
}_{H_{b}}
\begin{bmatrix} a_2(t) \\
	a_2(t)^\# 
\end{bmatrix}
dt \nonumber \\
&& \quad +
\underbrace{
\begin{bmatrix}
	K_{1} & K_{2} \\
	K_{2}^\# & K_{1}^\# \\
\end{bmatrix}
}_{K}
\begin{bmatrix}
	\dif y(t) \\
	\dif y(t)^\# 
\end{bmatrix}.
	\label{eqn:sysPerterbed}
\end{eqnarray} 

If the matrix $F_d$ is non-singular, it is possible to obtain a corresponding 
reduced dimension slow subsystem which we shall henceforth call the 
\emph{approximate system}, by formally setting $\varepsilon = 0$:
\begin{eqnarray}
\begin{bmatrix} \dif a_1(t) \\
	\dif a_1(t)^\# 
\end{bmatrix}
&=& 
F_0
\begin{bmatrix} a_1(t) \\
	a_1(t)^\# 
\end{bmatrix}
dt +
G_0
\begin{bmatrix}
	\dif u(t) \\
	\dif u(t)^\# 
\end{bmatrix};
\nonumber \\
\begin{bmatrix}
	\dif y(t) \\
	\dif y(t)^\# 
\end{bmatrix}
&=& 
H_0
\begin{bmatrix} a_1(t) \\
	a_1(t)^\# 
\end{bmatrix}
dt +
K_0
\begin{bmatrix}
	\dif u(t) \\
	\dif u(t)^\# 
\end{bmatrix}.
\label{eqn:sysApprox}
\end{eqnarray} 

where:
\begin{eqnarray}
	\label{eqn:F0}
	F_0 &=& F_a - F_b F_d{}^{-1} F_c; \nonumber \\
	G_0 &=& G_a - F_b F_d{}^{-1} G_b; \nonumber \\
	H_0 &=& H_a - H_b F_d{}^{-1} F_c; \nonumber \\
	K_0 &=& K   - H_b F_d{}^{-1} G_b. 
\end{eqnarray}

\begin{theorem}
	If the singularly perturbed linear complex quantum system 
	(\ref{eqn:sysPerterbed}) is physically realizable for all 
	$\varepsilon \ge 0$, and the matrix $F_d$ is non-singular,
	then the corresponding reduced dimension approximate 
	system (\ref{eqn:sysApprox}) 
	has transfer function matrix $\Phi_0(s) = H_0 (sI - F_0)^{-1}G_0 + K_0$ 
	such that $\Phi_0(s)^\dagger J \Phi_0(s) = J$ for all $s \in 
	\mathbb{C}$. That is, it is (J,J)-unitary.
\end{theorem}

\begin{remark}
	This result is not sufficient to prove the physical realizability of 
	approximate system, as Theorem 2 includes additional conditions 
	(with respect to minimality, eigenvalues, and the $K$ matrix) 
	that need to be met to ensure physical realizabilty. However, 
	these conditions can easily be checked for the approximate system, to 
	verify physical realizability.
\end{remark}

\begin{proof}
	Consider the transfer function of the singularly perturbed 
	system (\ref{eqn:sys}). From \cite{PET09A}:
	\begin{eqnarray}
		\Phi_\varepsilon(s) &=&  
		H(sI-F_\varepsilon)^{-1}G_\varepsilon + K \nonumber \\
		&=& \Phi_0(s) \nonumber \\
		&& {} - \varepsilon s \left(
		H_0(sI - F_0)^{-1} F_b + H_b F_d^{-1} 
		\right) F_d^{-1} \nonumber \\
		&& {}\times \left(
		F_c(sI - F_0)^{-1}G_0 + G_b \right)
		+ O(\varepsilon^2)
		\label{eqn:Phi_epsilon}
	\end{eqnarray}
	with variables defined as previously.

	If the singularly perturbed system (\ref{eqn:sys}) is physically 
	realizabile for all $\varepsilon > 0$, it follows from \cite{ShP5a} 
	that
	$$\Phi_\varepsilon(s)^\dagger J \Phi_\varepsilon(s) = J; \qquad 
	\forall s \in \mathbb{C},$$
	for all $\varepsilon > 0$. Hence, from  
	(\ref{eqn:Phi_epsilon}) if follows that
	$$\Phi_0(s)^\dagger J \Phi_0(s) = J $$ for all $s \in 
	\mathbb{C}$. That is, the approximate system is (J,J)-unitary.
\end{proof}

\subsection{A Special Class of Singular Perterbations}
We now turn our attention to a special class of singularly perturbed 
physically realizable quantum systems of the form (\ref{eqn:sys}) defined 
(as per Definition 1)
in terms of $M$, $N$, $S$ and canonical $\Theta$ as follows:
\begin{eqnarray}
	\Theta &=& J = 
	\begin{bmatrix}
		I & 0 \\ 0 & - I
	\end{bmatrix};
	\nonumber \\
	M &=& 
	\begin{bmatrix}
		M_{1a}    & 
		\frac{1}{\sqrt{\varepsilon}} M_{1b} & 
		M_{2a} &
		\frac{1}{\sqrt{\varepsilon}} M_{2b} \\
		\frac{1}{\sqrt{\varepsilon}} M_{1c}    & 
		\frac{1}{{\varepsilon}} M_{1d} &   
		\frac{1}{\sqrt{\varepsilon}} M_{2c}    & 
		\frac{1}{{\varepsilon}} M_{2d} \\
		M_{2a}^\# & 
		\frac{1}{\sqrt{\varepsilon}} M_{2b}^\# & 
		M_{1a}^\# & 
		\frac{1}{\sqrt{\varepsilon}} M_{1b}^\# \\
		\frac{1}{\sqrt{\varepsilon}} M_{2c}^\# & 
		\frac{1}{{\varepsilon}} M_{2d}^\# & 
		\frac{1}{\sqrt{\varepsilon}} M_{1c}^\# & 
		\frac{1}{{\varepsilon}} M_{1d}^\#
	\end{bmatrix}; \nonumber \\
	M &=& M^\dagger; \nonumber \\
	N &=& \begin{bmatrix}
		N_{1a}    & 
		\frac{1}{\sqrt{\varepsilon}} N_{1b} & 
		N_{2a} &
		\frac{1}{\sqrt{\varepsilon}} N_{2b} \\
		N_{2a}^\# &
		\frac{1}{\sqrt{\varepsilon}} N_{2b}^\# &
		N_{1a}^\#    & 
		\frac{1}{\sqrt{\varepsilon}} N_{1b}^\#  
	\end{bmatrix};
	\nonumber \\
	K &=& \begin{bmatrix} S & 0 \\ 0 & S^\# \end{bmatrix};
	\quad SS^\dagger = I.
	\label{eqn:spec}
\end{eqnarray}

For convenience define 
$M_a = \left[ \begin{smallmatrix} 
	M_{1a} & M_{2a} \\ M_{2a}^\# & M_{1a}^\# 
\end{smallmatrix} \right]$, and likewise for $M_b$, $M_c$, $M_d$, $N_a$ and
$N_b$.

From (\ref{eqn:spec}), we can obtain the system in the form (\ref{eqn:sys}) 
and thence of the form (\ref{eqn:sysPerterbed}). After the change of variables 
$\overline{a}_2(t) := \frac{1}{\sqrt{\varepsilon}}a_2(t)$, we obtain:
\begin{eqnarray}
	\label{eqn:specA}
\begin{bmatrix} \dif a_1(t) \\
	\dif a_1(t)^\# 
\end{bmatrix}
&=& 
\underbrace{
-J(i M_a + \frac{1}{2}N_a{}^\dagger J N_a)
}_{F_{a}}
\begin{bmatrix} a_1(t) \\
	a_1(t)^\# 
\end{bmatrix}
dt \nonumber \\
&& \quad \underbrace{
{} - J(i M_b + \frac{1}{2}N_a{}^\dagger J N_b)
}_{F_{b}}
\begin{bmatrix} \overline{a}_2(t) \\
	\overline{a}_2(t)^\# 
\end{bmatrix}
dt \nonumber \\
&& \quad \underbrace{
{} - J N_a{}^\dagger J K
}_{G_{a}}
\begin{bmatrix}
	\dif u(t) \\
	\dif u(t)^\# 
\end{bmatrix};
\nonumber \\
\varepsilon 
\begin{bmatrix} \dif \overline{a}_2(t) \\
	\dif \overline{a}_2(t)^\# 
\end{bmatrix}
&=& 
\underbrace{
-J(i M_c + \frac{1}{2}N_b{}^\dagger J N_a)
}_{F_{c}}
\begin{bmatrix} a_1(t) \\
	a_1(t)^\# 
\end{bmatrix}
dt \nonumber \\
&& \quad \underbrace{
{}-J(i M_d + \frac{1}{2}N_b{}^\dagger J N_b)
}_{F_{d}}
\begin{bmatrix} \overline{a}_2(t) \\
	\overline{a}_2(t)^\# 
\end{bmatrix}
dt \nonumber \\
&& \quad \underbrace{
{} - J N_b{}^\dagger J K
}_{G_{b}}
\begin{bmatrix}
	\dif u(t) \\
	\dif u(t)^\# 
\end{bmatrix};
\nonumber \\
\begin{bmatrix}
	\dif y(t) \\
	\dif y(t)^\# 
\end{bmatrix}
&=& 
\underbrace{N_a}_{H_{a}}
\begin{bmatrix} a_1(t) \\
	a_1(t)^\# 
\end{bmatrix}
dt +
\underbrace{N_b}_{H_{b}}
\begin{bmatrix} \overline{a}_2(t) \\
	\overline{a}_2(t)^\# 
\end{bmatrix}
dt \nonumber \\
&& \qquad + K
\begin{bmatrix}
	\dif y(t) \\
	\dif y(t)^\# 
\end{bmatrix}.
\nonumber \\
\end{eqnarray} 

Then from (\ref{eqn:F0}), we obtain the approximate system matrices:
\begin{eqnarray}
	F_0 &=& -J(i M_a + \frac{1}{2}N_a{}^\dagger J N_a) 
	+ J(i M_b + \frac{1}{2}N_a{}^\dagger J N_b)
	\nonumber \\ && \quad \times 
	(i M_d + \frac{1}{2}N_b{}^\dagger J N_b)^{-1} 
	(i M_c + \frac{1}{2}N_b{}^\dagger J N_a);
	\nonumber \\
	G_0 &=& - J N_a{}^\dagger J K
	+ J(i M_b + \frac{1}{2}N_a{}^\dagger J N_b)
	\nonumber \\ && \quad \times 
	(i M_d + \frac{1}{2}N_b{}^\dagger J N_b)^{-1} 
	N_b{}^\dagger J K;
	\nonumber \\
	H_0 &=& N_a 
	- N_b 
	(i M_d + \frac{1}{2}N_b{}^\dagger J N_b)^{-1} 
	(i M_c + \frac{1}{2}N_b{}^\dagger J N_a);
	\nonumber \\
	K_0 &=& K - N_b 
	(i M_d + \frac{1}{2}N_b{}^\dagger J N_b)^{-1} 
	N_b{}^\dagger J K.
	\label{eqn:specApprox}
\end{eqnarray}

\begin{theorem}
	Suppose the singularly perterbed system (\ref{eqn:spec}) is physically 
	realizable for all $\varepsilon > 0$. Then the approximate system 
	(\ref{eqn:specApprox}) is equivalent to a physically realizable system 
	connected in series with a static Bogoliubov component (static 
	squeezer).
\end{theorem}

The following lemma will be used in the proof of this result.

\begin{lemma} (See \cite{PET10Ba}.) A matrix R is of the form 
	$\left[ \begin{smallmatrix}
	R_1 & R_2 \\ R_2^\# & R_1^\# \end{smallmatrix} \right]$ if and only if 
	$R = \Sigma R^\# \Sigma$, where $\Sigma = \left[ 
	\begin{smallmatrix} 0 & I\\ I & 0 \end{smallmatrix} \right]$.
\end{lemma}

\begin{proof}
We wish to show that there exists $\tilde{M}$,$\tilde{N}$, and 
$\tilde{K}$ with
\begin{eqnarray*}
	\tilde{M} &=& \begin{bmatrix} \tilde{M}_1 & \tilde{M}_2 \\
		\tilde{M}_2^\# & \tilde{M}_1^\# \end{bmatrix}; 
		\qquad \tilde{M}^\dagger = \tilde{M}; \\
	\tilde{N} &=& \begin{bmatrix} \tilde{N}_1 & \tilde{N}_2 \\
		\tilde{N}_2^\# & \tilde{N}_1^\# \end{bmatrix};  \\
	\tilde{K} &=& \begin{bmatrix} \tilde{K}_1 & \tilde{K}_2 \\
		\tilde{K}_2^\# & \tilde{K}_1^\# \end{bmatrix}; 
		\qquad \tilde{K}^{-1} = J \tilde{K}^\dagger J ,
\end{eqnarray*}
such that,
\begin{eqnarray*}
	F_0 &=& -iJ\tilde{M} - \frac{1}{2}J\tilde{N}^\dagger J\tilde{N}; \\
	G_0 &=& -J \tilde{N}^\dagger J \tilde{K}; \\
	H_0 &=& \tilde{N}; \\
	K_0 &=& \tilde{K}.
\end{eqnarray*}

Let
\begin{eqnarray}
	\tilde{M} &=& M_a  \nonumber \\
	&& {} - i \frac{1}{2} M_b 
	(i M_d + \frac{1}{2}N_b^\dagger J N_b)^{-1} M_c \nonumber \\
	&& {} + i \frac{1}{2} M_b 
	(-i M_d + \frac{1}{2}N_b^\dagger J N_b)^{-1} M_c \nonumber \\
	&& {} - \frac{1}{4} M_b (i M_d + \frac{1}{2}N_b^\dagger J N_b)^{-1} 
	N_b^\dagger J N_a \nonumber \\
	&& {} - \frac{1}{4} M_b (-i M_d + \frac{1}{2}N_b^\dagger J N_b)^{-1} 
	N_b^\dagger J N_a \nonumber \\
	&& {} - \frac{1}{4} N_a^\dagger J N_b  
	(i M_d + \frac{1}{2}N_b^\dagger J N_b)^{-1} M_c \nonumber \\
	&& {} - \frac{1}{4} N_a^\dagger J N_b 
	(-i M_d + \frac{1}{2}N_b^\dagger J N_b)^{-1} M_c \nonumber \\
	&& {} +i \frac{1}{8} N_a^\dagger J N_b  
	(i M_d + \frac{1}{2}N_b^\dagger J N_b)^{-1} N_b^\dagger J N_a 
	\nonumber \\
	&& {} -i \frac{1}{8} N_a^\dagger J N_b 
	(-i M_d + \frac{1}{2}N_b^\dagger J N_b)^{-1} N_b^\dagger J N_a 
	\nonumber \\
	\tilde{N} &=& N_a 
	- N_b 
	(i M_d + \frac{1}{2}N_b{}^\dagger J N_b)^{-1} 
	(i M_c + \frac{1}{2}N_b{}^\dagger J N_a)
	\nonumber \\
	\tilde{K} &=& K - N_b 
	(i M_d + \frac{1}{2}N_b{}^\dagger J N_b)^{-1} 
	N_b{}^\dagger J K
	\label{eqn:specSoln}
\end{eqnarray}

Consider $K_0 =\tilde{K}$ first. In order to apply Lemma 1, 
we wish to show that $\Sigma \tilde{K}^\# 
\Sigma = \tilde{K}$, and that $J \tilde{K}^\dagger J \tilde{K} = I$. Indeed,
\begin{eqnarray*}
	\Sigma \tilde{K}^\# \Sigma &=& \Sigma K^\# \Sigma 
	- \Sigma N_b^\# \Sigma 
	\left( -i\Sigma M_d^\# \Sigma  \right. 
	\\ && 
	\qquad \left. {} + \frac{1}{2} 
	(\Sigma N_b{}^\# \Sigma)^\dagger \Sigma J \Sigma 
	\Sigma N_b^\# \Sigma \right) ^{-1} \\
	&& \qquad {} \times 
	(\Sigma N_b{}^\# \Sigma )^\dagger \Sigma J \Sigma 
	\Sigma K^\# \Sigma \\
	&=& K - N_b \left( -i M_d + \frac{1}{2} N_b{}^\dagger (-J) N_b
	\right)^{-1} 
	\\ && \qquad {} \times
	N_b{}^\dagger (- J) K \\
	&=& K - N_b (i M_d + \frac{1}{2} N_b{}^\dagger J N_b)^{-1} 
	N_b{}^\dagger J K \\
	&=& \tilde{K}; \\
\end{eqnarray*}

\begin{eqnarray*}
	\lefteqn{J \tilde{K}^\dagger J \tilde{K} =} \\
	&& J \left( K^\dagger - K^\dagger J N_b 
	(-i M_d + \frac{1}{2}N_b{}^\dagger J N_b)^{-1} 
	N_b{}^\dagger \right) \\
	&& \qquad \times
	J
	\left( K - N_b 
	(i M_d + \frac{1}{2}N_b{}^\dagger J N_b)^{-1} 
	N_b{}^\dagger J K \right) \\
&=& 
	J K^\dagger
	\left( I - J N_b 
	(-i M_d + \frac{1}{2}N_b{}^\dagger J N_b)^{-1} 
	N_b{}^\dagger \right) \\
	&& \qquad \times
	\left( I - J N_b 
	(i M_d + \frac{1}{2}N_b{}^\dagger J N_b)^{-1} 
	N_b{}^\dagger \right) JK \\
&=& 
	JK^\dagger \left( I
	- J N_b 
	(-i M_d + \frac{1}{2}N_b{}^\dagger J N_b)^{-1} 
	N_b{}^\dagger \right. \\
	&& \qquad - J N_b 
	(i M_d + \frac{1}{2}N_b{}^\dagger J N_b)^{-1} 
	N_b{}^\dagger \\
	&& \qquad + J N_b 
	(-i M_d + \frac{1}{2}N_b{}^\dagger J N_b)^{-1} 
	N_b{}^\dagger 
	J N_b \\
	&& \left. \qquad \times
	(i M_d + \frac{1}{2}N_b{}^\dagger J N_b)^{-1} 
	N_b{}^\dagger 
	\right) JK \\
&=& 
	JK^\dagger JK \\
	&& \qquad - JK^\dagger J N_b \left(
	(-i M_d + \frac{1}{2}N_b{}^\dagger J N_b)^{-1} 
	\right.
	\\ &&
	\qquad {} + (i M_d + \frac{1}{2}N_b{}^\dagger J N_b)^{-1} 
	\\ && 
	\qquad {} - (-i M_d + \frac{1}{2}N_b{}^\dagger J N_b)^{-1} 
	N_b^\dagger J N_b 
	\\ &&
	\qquad \left. {} \times (i M_d + \frac{1}{2}N_b{}^\dagger J N_b)^{-1} 
	\right)	N_b{}^\dagger JK \\
&=& 
	JK^\dagger JK \\
	&& \qquad - JK^\dagger J N_b 
	(-i M_d + \frac{1}{2}N_b{}^\dagger J N_b)^{-1} 
	\\ && 
	\qquad \times \left(
	i M_d + \frac{1}{2}N_b{}^\dagger J N_b -i M_d \right.
	\\ && 
	\qquad \left. {} + \frac{1}{2}N_b{}^\dagger J N_b
	- N_b^\dagger J N_b \right) \\
	&& \qquad \times (i M_d + \frac{1}{2}N_b{}^\dagger J N_b)^{-1} 
	N_b{}^\dagger JK \\
	&=& JK^\dagger JK \\
	&=& 
	\begin{bmatrix} I & 0 \\ 0 & -I \end{bmatrix}
	\begin{bmatrix} S^\dagger & 0 \\ 0 & S^T \end{bmatrix}
	\begin{bmatrix} I & 0 \\ 0 & -I \end{bmatrix}
	\begin{bmatrix} S & 0 \\ 0 & S^\# \end{bmatrix} \\
	&=& \begin{bmatrix} S^\dagger S & 0 \\ 0 & S^T S^\# \end{bmatrix} \\
	&=& I.
\end{eqnarray*}

We now consider $\tilde{N} = H_0$. As with $\tilde{K}$ it is straightforward to 
show that $\Sigma \tilde{N}^\# \Sigma = \tilde{N}$. We wish to show that
$G_0 = -J \tilde{N}^\dagger JK$, i.e. that 
$G_0  + J \tilde{N}^\dagger JK = 0$. Indeed,

\begin{eqnarray*}
	\lefteqn{G_0 + J \tilde{N}^\dagger JK =} 
	\\ &&  
	{} - J K_a{}^\dagger J K
	+ J(i M_b + \frac{1}{2}N_a{}^\dagger J N_b)
	\\ &&
	\qquad \times (i M_d + \frac{1}{2}N_b{}^\dagger J N_b)^{-1} 
	N_b{}^\dagger J K 
	\\ &&  
	+ J \left(N_a^\dagger - 
	(-i M_b + \frac{1}{2}N_a{}^\dagger J N_b) \right.
	\\ &&
	\qquad \times \left.
	(-i M_d + \frac{1}{2}N_b{}^\dagger J N_b)^{-1} 
	N_b^\dagger
	\right)  
	\\ && 
	\qquad \times J \left( K - N_b 
	(i M_d + \frac{1}{2}N_b{}^\dagger J N_b)^{-1} 
	N_b{}^\dagger J K \right) 
\end{eqnarray*}

\begin{eqnarray*}
&=& 
	{} - J N_a{}^\dagger J K 
	\\ && 
	{}+ J i M_b 
	(i M_d + \frac{1}{2}N_b{}^\dagger J N_b)^{-1} 
	N_b{}^\dagger J K 
	\\ && 
	{} + J \frac{1}{2}N_a{}^\dagger J N_b
	(i M_d + \frac{1}{2}N_b{}^\dagger J N_b)^{-1} 
	N_b{}^\dagger J K 
	\\ && 
	{} + J N_a^\dagger JK 
	\\ && 
	{} - J N_a^\dagger J N_b 
	(i M_d + \frac{1}{2}N_b{}^\dagger J N_b)^{-1} 
	N_b{}^\dagger J K 
	\\ && 
	{} + J i M_b
	(-i M_d + \frac{1}{2}N_b{}^\dagger J N_b)^{-1} 
	N_b^\dagger J K 
	\\ &&
	{} - J \frac{1}{2}N_a{}^\dagger J N_b
	(-i M_d + \frac{1}{2}N_b{}^\dagger J N_b)^{-1} 
	N_b^\dagger J K 
	\\ && 
	{} - J i M_b
	(-i M_d + \frac{1}{2}N_b{}^\dagger J N_b)^{-1} 
	N_b^\dagger J N_b
	\\ &&
	\qquad \times
	(i M_d + \frac{1}{2}N_b{}^\dagger J N_b)^{-1} 
	N_b{}^\dagger J K 
	\\ && 
	{} + J \frac{1}{2}N_a{}^\dagger J N_b
	(-i M_d + \frac{1}{2}N_b{}^\dagger J N_b)^{-1} 
	N_b^\dagger J 
	\\ &&
	\qquad \times
	N_b (i M_d + \frac{1}{2}N_b{}^\dagger J N_b)^{-1} 
	N_b{}^\dagger J K 
\end{eqnarray*}

\begin{eqnarray*}
&=& 
	+ J i M_b 
	(i M_d + \frac{1}{2}N_b{}^\dagger J N_b)^{-1} 
	N_b{}^\dagger J K 
	\\ && 
	{} + J i M_b
	(-i M_d + \frac{1}{2}N_b{}^\dagger J N_b)^{-1} 
	N_b^\dagger J K 
	\\ && 
	{} - J i M_b
	(-i M_d + \frac{1}{2}N_b{}^\dagger J N_b)^{-1} 
	N_b^\dagger J N_b 
	\\ &&
	\qquad \times
	(i M_d + \frac{1}{2}N_b{}^\dagger J N_b)^{-1} 
	N_b{}^\dagger J K 
	\\ && 
	{} - J \frac{1}{2}N_a{}^\dagger J N_b
	(i M_d + \frac{1}{2}N_b{}^\dagger J N_b)^{-1} 
	N_b{}^\dagger J K 
	\\ && 
	{} - J \frac{1}{2}N_a{}^\dagger J N_b
	(-i M_d + \frac{1}{2}N_b{}^\dagger J N_b)^{-1} 
	N_b^\dagger J K 
	\\ && 
	{} + J \frac{1}{2}N_a{}^\dagger J N_b
	(-i M_d + \frac{1}{2}N_b{}^\dagger J N_b)^{-1} 
	\\ &&
	\qquad \times
	N_b^\dagger J 
	N_b (i M_d + \frac{1}{2}N_b{}^\dagger J N_b)^{-1} 
	N_b{}^\dagger J K	
\end{eqnarray*}

\begin{eqnarray*}
	&=& 
	{} + J i M_b \left( 
	(i M_d + \frac{1}{2}N_b{}^\dagger J N_b)^{-1} 
	\right. 
	\\ && 
	\qquad {} + (-i M_d + \frac{1}{2}N_b{}^\dagger J N_b)^{-1} 
	\\ && 
	\qquad {} - (-i M_d + \frac{1}{2}N_b{}^\dagger J N_b)^{-1} 
	N_b^\dagger J N_b 
	\\ &&
	\qquad \qquad \times
	\left. (i M_d + \frac{1}{2}N_b{}^\dagger J N_b)^{-1} 
	\right) N_b{}^\dagger J K 
	\\ &&
	{} - J \frac{1}{2}N_a{}^\dagger J N_b i
	\\ &&
	{} \times \left(
	(i M_d + \frac{1}{2}N_b{}^\dagger J N_b)^{-1} \right. 
	- (-i M_d + \frac{1}{2}N_b{}^\dagger J N_b)^{-1} 
	\\ && 
	\qquad {} + (-i M_d + \frac{1}{2}N_b{}^\dagger J N_b)^{-1} 
	\\ &&
	\qquad \times
	\left. N_b^\dagger J 
	N_b (i M_d + \frac{1}{2}N_b{}^\dagger J N_b)^{-1} 
	\right) N_b{}^\dagger J K \\
&=& 0.
\end{eqnarray*}

Finally, we set
$$ \tilde{M} = iJ \left( F_0 + \frac{1}{2} J \tilde{N}^\dagger 
J \tilde{N} \right). $$

After simplification along similar lines to that shown for the previous 
equations, this yields the expression for $\tilde{M}$ given in 
(\ref{eqn:specSoln}) which it is straightforward to verify is hermitian. 

This completes the proof of the theorem.
\end{proof}

\section{Illustrative Example}\label{sec:ex}
The following example from quantum optics demonstrates our main result. The 
example is similar to that in \cite{PET09A}. However, here we consider a 
cavity coupled to a squeezer as shown in Figure 1. Unlike in \cite{PET09A}, 
the evolution of this system is in terms of both annihilation and creation 
operators.

\begin{figure}[h]
\psfrag{K1}{$K_1$}
\psfrag{K2}{$K_2$}
\psfrag{u1}{$u_1$}
\psfrag{u2}{$u_2$}
\psfrag{y1}{$y_1$}
\psfrag{y2}{$y_2$}
\psfrag{a1}{$a_1$}
\psfrag{a2}{$a_2$}
\psfrag{gamma}{$\gamma$}
\includegraphics[trim = 0mm 0mm 0mm -10mm, scale=0.8]{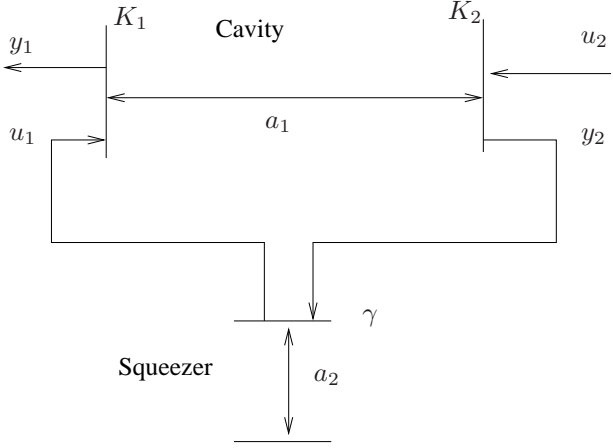}
\caption{A linear optical quantum system.}
\label{fig:1}
\end{figure}

Here,
$K_1$ and $K_2$ are the coupling parameters of the first cavity, $\gamma$ is 
the coupling parameter of the squeezer and $\chi$ is the squeezing parameter. 
The system under consideration can be described by QSDE of the 
form (\ref{eqn:genModel}) with:
\begin{eqnarray}
F &=& 
\left[ \begin{smallmatrix}
	- \frac{1}{2}(\sqrt{K_1} + \sqrt{K_2})^2 & 
	- \sqrt{K_1 \gamma} & 
	0 & 
	0 \\
	- \sqrt{K_2 \gamma} & 
	- \frac{\gamma}{2} & 
	0 & 
	- \chi \\
	0 & 
	0 & 
	- \frac{1}{2}(\sqrt{K_1} + \sqrt{K_2})^2 & 
	- \sqrt{K_1 \gamma} \\ 
	0 & 
	\chi^* &
	- \sqrt{K_2 \gamma} & 
	- \frac{\gamma}{2} \\ 
\end{smallmatrix} \right]; \nonumber \\
G &=& 
\left[ \begin{smallmatrix}
	- \sqrt{K_1} - \sqrt {K_2} & 0 \\
	- \sqrt{\gamma} & 0 \\
	0 & - \sqrt{K_1} - \sqrt {K_2} \\
	0 & - \sqrt{\gamma} \\
\end{smallmatrix} \right];
\nonumber \\
H &=& 
\left[ \begin{smallmatrix}
	(\sqrt{K_1} + \sqrt {K_2}) & 
	\sqrt{\gamma} & 0 & 0 \\
	0 & 0 & (\sqrt{K_1} + \sqrt {K_2}) & 
	\sqrt{\gamma} \\
\end{smallmatrix} \right];
\nonumber \\
K &=& 
\left[ \begin{smallmatrix}
	1 & 0 \\
	0 & 1 \\
\end{smallmatrix} \right]. 
\label{eqn:ex1}
\end{eqnarray} 

Now suppose that the dynamics of the squeezer are at a much higher frequency 
than that of interest. We can apply a singular perterbation approximation by 
letting $\gamma = \frac{1}{\varepsilon}\tilde{\gamma}$ and
$\chi = \frac{1}{\varepsilon}\tilde{\chi}$. After the change of variables
$\overline{a}_2 = \frac{1}{\varepsilon}a_2$ we obtain a system of the form 
(\ref{eqn:sys}). In fact, this system belongs to the special class in which 
$M$, $N$, $S$, and $\Theta$ are of the form (\ref{eqn:spec}) as follows:

\begin{eqnarray*}
	M_a &=&  \left[ \begin{smallmatrix}
		0 & 0 \\ 0 & 0 
	\end{smallmatrix} \right]; \nonumber \\
	M_b &=&  \left[ \begin{smallmatrix}
		\frac{i}{2}(\sqrt{K_1} - \sqrt{K_2})\gamma &
		0 \\ 0 &
		-\frac{i}{2}(\sqrt{K_1} - \sqrt{K_2})\gamma
	\end{smallmatrix} \right]; \nonumber \\
	M_c &=&  \left[ \begin{smallmatrix}
		- \frac{i}{2}(\sqrt{K_1} - \sqrt{K_2})\gamma &
		0 \\ 0 &
		\frac{i}{2}(\sqrt{K_1} - \sqrt{K_2})\gamma
	\end{smallmatrix} \right]; \nonumber \\
	M_d &=&  \left[ \begin{smallmatrix}
		0 & i \chi \\ -i \chi^* & 0
	\end{smallmatrix} \right]; \nonumber \\
	N_a &=&  \left[ \begin{smallmatrix}
		(\sqrt{K_1} + \sqrt{K_2}) &
		0 \\ 0 &
		(\sqrt{K_1} + \sqrt{K_2})
	\end{smallmatrix} \right]; \nonumber \\
	N_b &=&  \left[ \begin{smallmatrix}
		\sqrt{\gamma} &
		0 \\ 0 &
		\sqrt{\gamma}
	\end{smallmatrix} \right]; \nonumber \\
	S &=& I.
\end{eqnarray*}

Applying the singular perturbation approximation, we obtain a system 
of the form (\ref{eqn:sysApprox}), and from (\ref{eqn:specApprox}) we 
have:
\begin{eqnarray}
	F_0 &=& - \frac{1}{2}(\sqrt{K_1} + \sqrt{K_2})^2 I +
		\gamma \sqrt{K_1 K_2} \left[ \begin{smallmatrix}
			\frac{\gamma}{2} & - \chi \\ 
			- \chi^* & \frac{\gamma}{2}
	\end{smallmatrix} \right]^{-1} \nonumber \\
	G_0 &=& - \frac{1}{2}(\sqrt{K_1} + \sqrt{K_2}) I +
		\gamma \sqrt{K_2} \left[ \begin{smallmatrix}
			\frac{\gamma}{2} & - \chi \\ 
			- \chi^* & \frac{\gamma}{2}
	\end{smallmatrix} \right]^{-1} \nonumber \\
	H_0 &=& \frac{1}{2}(\sqrt{K_1} + \sqrt{K_2}) I -
		\gamma \sqrt{K_1} \left[ \begin{smallmatrix}
			\frac{\gamma}{2} & - \chi \\ 
			- \chi^* & \frac{\gamma}{2}
	\end{smallmatrix} \right]^{-1} \nonumber \\
	K_0 &=& I - \gamma \left[ \begin{smallmatrix}
			\frac{\gamma}{2} & - \chi \\ 
			- \chi^* & \frac{\gamma}{2}
	\end{smallmatrix} \right]^{-1} \label{eqn:ex3} 
\end{eqnarray}

Finally, since the system described by (\ref{eqn:ex1}) satisfies the 
conditions for Theorem 3, the approximate system described by (\ref{eqn:ex3}) 
is equivalent to a physically realizable system in series with a static 
Bogoliubov component (static squeezer).

This can be verified by obtaining $\tilde{M}$, $\tilde{N}$ and 
$\tilde{K}$ from (\ref{eqn:specSoln}) and confirming that the physically 
realizable system obtained from substituting $M = \tilde{M}$, $N = \tilde{N}$ 
and $S = I$ into (\ref{eqn:realizable}) and the static Bogoliubov 
component described by $\tilde{K}$ combine to form (\ref{eqn:ex3}).

\section{Conclusion} \label{sec:conc}

In this paper, we have considered singular perturbation approximations for the 
general class of quantum linear systems described by both annihilation and 
creation operators. Two main results were presented. We first considered 
a general singular perturbation approximation and obtained a result 
(relevant to physical realizability) relating to the J-J unitary property of 
the transfer function of the approximate system. We then considered the special 
case in which the Hamiltonian and coupling operators are singularly perturbed. 
While in general the system obtained from the singular perturbation 
approximation for the special case is not necessarily physically realizable, 
it is equivalent to a physically realizable system in series with a 
static Bogoliubov component (generalized static squeezer).

\bibliography{irpnew}  
\bibliographystyle{IEEEtran}

\end{document}